\documentclass[11pt,a4paper]{article}

\usepackage{a4wide}
\usepackage{amsmath,amsthm,amssymb,hyperref}
\usepackage{filecontents} 
\usepackage{algorithm}
\usepackage{algorithmic} 
\hypersetup{colorlinks=true,citecolor=blue, linkcolor=blue, urlcolor=blue}

\makeatletter
\def\prob#1#2#3{\goodbreak\begin{list}{}{\labelwidth\z@ \itemindent-\leftmargin
                        \itemsep\z@  \topsep6\p@\@plus6\p@
                        \let\makelabel\descriptionlabel}
                \item[\it Name.]#1
               \item[\it Instance.]                #2
                \item[\it Output.]#3
                \end{list}}
\makeatother 

\newtheorem{theorem}{Theorem}
\newtheorem{lemma}[theorem]{Lemma}

\theoremstyle{definition}
\newtheorem{definition}[theorem]{Definition}
\newtheorem{observation}[theorem]{Observation}
\newtheorem{remark}[theorem]{Remark}

\newtheorem*{thmmain}{Theorem \ref{thm:main}}
 
\newcommand{\Ex}{\mathop{\mathbb{{}E}}\nolimits} 

\newcommand{\ZI}{Z^\mathrm{Ising}_{G,\bbeta}}
\newcommand{\ZMonoI}{Z^\mathrm{Ising}_{G,\sbeta}}

\newcommand{\ZMonoIG}[1]{Z^\mathrm{Ising}_{#1,\sbeta}}
\newcommand{\ZIstGb}[4]{Z^{\text{ Ising}}_{#3,#4,#1,#2}}
 
\newcommand{\ZE}[1]{Z_{#1}(G;\blambda)}
\newcommand{\hatZS}[1]{\widehat Z_{#1}(G;\blambda,w)}
\newcommand{\hatZ}{\widehat Z(G;\blambda,w)}
\newcommand{\hatZp}[2]{\widehat Z_{#1}(#2)}

\newcommand{\Zrc}{Z^\mathrm{RC}_{G;\bp}}

\newcommand{\vdeg}{\mathop{\mathrm{deg}}}
\newcommand{\wt}{\mathrm{wt}}
\newcommand{\wtrc}{\wt^\mathrm{RC}_{G,\bp}}
\newcommand{\wtI}{\wt^\mathrm{Ising}_{G,\bbeta}}
\newcommand{\wtMonoI}{\wt^\mathrm{Ising}_{G,\sbeta}}
\newcommand{\piI}{\pi^\mathrm{Ising}_{G,\bbeta}}
\newcommand{\piIprime}{\pi^\mathrm{Ising}_{G',\bbeta}}
\newcommand{\piMonoI}[1]{\pi^\mathrm{Ising}_{#1,\sbeta}}
\newcommand{\piMonoIbeta}[2]{\pi^\mathrm{Ising}_{#1, #2}}
\newcommand{\pirc}{\pi^\mathrm{RC}_{G,\bp}}
\newcommand{\Omegabar}{\overline\Omega_\emptyset}
\newcommand{\Omegahat}{\widehat\Omega}
\newcommand{\cp}[2]{\gamma(#1,#2)}
\newcommand\cps{\mathrm{cp}}
\newcommand\enc{\eta_{T,T'}}
\newcommand{\calE}{\mathcal{E}}

\def\ourrange{\mathbb{Q}}
\def\ferrorange{\mathbb{Q}_{>1}}
\def\calW{\mathcal{W}} 

\def\acc{1.1}  
\def\bd{\tfrac{1}{2m}}
\def\KK{8}

\def\bbeta{ \beta}  
\def\sbeta{ b}  
\def\blambda{ \lambda}  
\def\bp{ p}  
\let\epsilon=\varepsilon

\def\calM{\mathcal{M}}
\def\Omegastar{\Omega^*}
\def\pistar{\pi^*}  
\newcommand\tmix[1]{t_{\mathrm{mix},#1}}

\def\acc{\epsilon}  
\def\goal{\hat{\sbeta}}
\def\myb{\sbeta}
 
\def\nmin{\nu_{\text{min}}}
\def\nmax{\nu_{\text{max}}}

\def\FerroIsingCorr{\ensuremath{\mathsf{FerroIsingCov}}}
\def\IsingCorr{\ensuremath{\mathsf{IsingCov}}}
\def\Correlation{\ensuremath{\mathsf{SignIsingCov}_{\myb}}}

\def\Pr{\mathop{\rm Pr}\nolimits}
\def\calD{\mathcal{D}}
\def\DSW{\calD_{G,\bp}}

\newcommand{\calF}{\mathcal F}
\let\rho=\varrho
\newcommand{\lambdamin}{\lambda_{\min}}
\newcommand{\lambdaimin}{\lambda^{[i]}_{\min}}

\def\hatR{\widehat{R}}
\def\hatQ{\widehat{Q}}

\title {Approximating Pairwise Correlations in the Ising Model\thanks{To Appear in ACM ToCT}}
\author{Leslie Ann Goldberg\thanks{University of Oxford,  UK. The research leading to these results has received funding from the European Research Council under the European Union's Seventh Framework Programme (FP7/2007--2013) ERC grant agreement no.\ 334828. The paper reflects only the authors' views and not the views of the ERC or the European Commission. The European Union is not liable for any use that may be made of the information contained therein.}
\and Mark Jerrum\thanks{Queen Mary, University of London, UK. Supported by EPSRC grant EP/N004221/1.}}

\date{25 April 2019}

\begin{document}
\maketitle{}

\begin{abstract}
In the Ising model, we consider the problem of estimating the covariance of the spins at two specified vertices.
In the ferromagnetic case, it is easy to obtain an additive approximation to this covariance by repeatedly sampling from
the relevant Gibbs distribution.  However, we desire a multiplicative approximation, and it is not clear how to achieve
this by sampling, given that the covariance can be exponentially small.
Our main contribution is a fully polynomial time randomised approximation scheme (FPRAS) for the covariance in the ferromagnetic case. 
We also show that that the restriction to the ferromagnetic case is essential --- there is no FPRAS for
multiplicatively estimating the covariance of an antiferromagnetic Ising model unless RP = \#P.
In fact, we show that even determining the sign of the covariance is \#P-hard in the antiferromagnetic case.
\end{abstract}

\section{Introduction}\label{sec:intro}

Let $G=(V,E)$ be a  
graph and let $\bbeta:E\to \ourrange$ be an edge weighting of~$G$.  
 A \emph{configuration} of the Ising model is an assignment $\sigma:V\to\{-1,+1\}$ of  \emph{spins}
 from $\{-1,+1\}$
  to the vertices of $G$. 
The weight of a configuration is 
$$\wtI(\sigma)= \prod_{\substack{e=\{u,v\}\in E:\\ \sigma(u)=\sigma(v)}}\bbeta(e).$$
The Ising partition function is
$\ZI=\sum_{\sigma:V\to\{-1,+1\}}\wtI(\sigma)$.
It is the normalising factor that makes the weights of configurations into a probability distribution, 
$\piI(\cdot)$,
which is called the \emph{Gibbs distribution} of the Ising model.
Thus, the probability of observing configuration~$\sigma$ is $\piI(\sigma)=\wtI(\sigma)/\ZI$.

We say that an edge weighting is   \emph{ferromagnetic} if 
$\bbeta(e)>1$ for all $e\in E$. The corresponding Ising model is also said to be ferromagnetic in this case.
We say that an edge weighting and the corresponding Ising model are
\emph{antiferromagnetic} if $0<\bbeta(e)<1$ for all $e\in E$.

Given specified vertices $s,t$, we are interested in computing $\Ex_{\piI}[\sigma(s)\sigma(t)]$.
Since $$\Ex_{\piI}[\sigma(s)]=\Ex_{\piI}[\sigma(t)]=0,$$ this quantity is equal to the covariance of the spins at $s$ and~$t$.

Interestingly, none of the existing work on computational aspects of the ferromagnetic Ising model provides an efficient algorithm for
estimating this covariance.
Jerrum and Sinclair~\cite{JS93} presented a polynomial-time algorithm for approximating the partition function $\ZI$ 
within specified relative error, 
and Randall and Wilson~\cite{RW99} observed that this algorithm could be used to produce samples from the Gibbs distribution.
Therefore, by repeated sampling we can easily get an additive approximation to the covariance.  Specifically, the covariance may be estimated to additive error $\epsilon$ using $O(\epsilon^{-2})$ samples.  

Our main contribution (Theorem~\ref{thm:main}) is a polynomial-time algorithm to approximate the covariance within small multiplicative error.
This is much more challenging than obtaining an additive approximation since the covariance may
be exponentially small in $n$, as will typically be the case when the system is in the uniqueness regime.
The computational problem that we study is the following.

\prob
{$\FerroIsingCorr$.} 
{ A   graph $G=(V,E)$  with 
specified vertices $s$ and $t$. An edge weighting $\bbeta:E\to \ferrorange$ of $G$.
 } 
{$ \Ex_{\piI}[\sigma(s)\sigma(t)]$.}

\newcommand{\statethmmain}{There is an FPRAS for $\FerroIsingCorr$.}

The reason that we restrict the range of the edge weighting~$\beta$
to the rationals 
(rather than allowing real-valued weights)
is to avoid the issue of how to represent real numbers in the input.
Each weight $\beta(e)$ satisfies $\beta(e)>1$.
For concreteness, we assume that  it is represented in the input by two
positive integers $P(e)$ and $Q(e)$ (specified in unary\footnote{The assumption that 
$P(e)$ and $Q(e)$ are specified in unary is a technical simplification, but is not essential:
see Remark~\ref{rem:unarybinary}.} in the input)
such that  
 and $\beta(e) = 1+P(e)/Q(e)$.
Our main result is that there is a polynomial-time approximation algorithm 
for $\FerroIsingCorr$.
In order to state the result precisely, we need to recall a definition from computational complexity. 
We view a problem, such as $\FerroIsingCorr$, as a function $f:\Sigma^\ast\to\ourrange$ from problem instances to rational numbers.

\begin{definition}
A \emph{randomised approximation scheme\/} for~$f:\Sigma^\ast\to\ourrange$ is a
randomised algorithm that takes as input an instance $ x\in
\Sigma^{\ast }$ (e.g., an encoding of a labelled graph) and an error
tolerance $\varepsilon >0$, and outputs a
number $z\in\ourrange$
(a random variable on the ``coin tosses'' made by the algorithm)
such that, for every instance~$x$,
$$
\Pr \bigg[e^{-\epsilon} \leq \frac{z}{f(x)} \leq
e^\epsilon \bigg]\geq \frac{3}{4}\,,
$$
where, by convention, $0/0 = 1$.
The randomised approximation scheme is said to be a
\emph{fully polynomial randomised approximation scheme},
or \emph{FPRAS},
if it runs in time bounded by a polynomial
in $ |x| $ and $ \epsilon^{-1} $.
(See Mitzenmacher and Upfal~\cite[Definition 10.2]{MU05}.)
\end{definition}
The slight modification of the more familiar definition is to ensure that functions $f$ taking 
negative values are dealt with correctly.

\begin{theorem}\label{thm:main} 
\statethmmain
\end{theorem}

 The restriction to the ferromagnetic case in Theorem~\ref{thm:main} is crucial.
Consider the unrestricted version of the problem.
 \prob
{$\IsingCorr$.} 
{ A   graph $G=(V,E)$  with 
specified vertices $s$ and $t$. An edge weighting $\bbeta:E\to \ourrange_{>0}$ of~$G$.
 } 
{$ \Ex_{\piI}[\sigma(s)\sigma(t)]$.}

We show the following.
\begin{theorem}\label{thm:negthm}
There is no FPRAS for $\IsingCorr$ unless $\mathrm{RP}=\mathrm{\#P}$.
\end{theorem}

Theorem~\ref{thm:negthm} holds even in the restricted setting 
where, for some fixed $b\in(0,1)$, the edge weighting~$\beta$ is the constant function which assigns every edge weight $\beta(e)=b$.
Theorem~\ref{thm:anti}
in Section~\ref{sec:antiferro}
shows that even showing whether 
$ \Ex_{\piI}[\sigma(s)\sigma(t)]$ 
is at least~$0$ or at most~$0$
is \#P-hard, in this restricted setting.
Theorem~\ref{thm:negthm} is an immediate consequence of Theorem~\ref{thm:anti}.

In Section~\ref{sec:proof} we prove
Theorem~\ref{thm:main}  
by providing an FPRAS for $\FerroIsingCorr$. Our FPRAS  is based on Markov-chain simulation.
Like the known MCMC algorithms for approximating the partition function of the Ising model,
it is explained in terms of a related model called the even subgraphs model. 
Our Markov chain  is a modification of a process known as the worm process.

\section{The even subgraphs model and the worm process}\label{sec:even}

An instance of the even subgraphs model is a graph $G=(V,E)$
with an edge weighting  $\blambda:E\to\ourrange_{>0}$.
A configuration of the model is a subset $A\subseteq E$
such that every vertex in the subgraph $(V,A)$ has even degree.

\begin{definition}\label{def:lamA}
We use the notation $\blambda(A)$ to denote the product $\blambda(A) = \prod_{e\in A} \blambda(e)$
of edge-weights of the edges in~$A$.
\end{definition}

It is convenient to generalise the even subgraphs model to allow a small set $S\subseteq V$ of ``exceptional vertices'' of odd degree.   
The configuration space of the (extended) even subgraphs model is given by
$$
\Omega_S=\big\{A\subseteq E :\text{ $\vdeg(v)$ is odd in $(V,A)$ iff $v\in S$}\big\},
$$
and the corresponding partition function is given by 
$$
\ZE{S}=\sum_{A\in\Omega_S}\blambda(A).
$$
Despite appearances, there is a close connection between the Ising model and the even subgraphs model.
Suppose that, for every $e\in E$,
$\blambda(e)=(\bbeta(e)-1)/(\bbeta(e)+1)$. Van der Waerden~\cite{vdW} showed that
there is an easily-computable scaling factor~$C$ such that
$\ZI=C\,\ZE\emptyset$.  
Note that a ferromagnetic Ising model corresponds to an even-subgraphs model
in which   $0<\blambda(e)<1$ for all $e\in E$.
We do not use  precisely van der Waerden's identity, but we do use a closely related one
which is captured by the following lemma, which can be found, e.g., in \cite[Lemma 2.1]{CGHT}.

\begin{lemma}\label{lem:worm} 
Let $G=(V,E)$ be a graph with edge weighting $\bbeta$.
Let $\blambda$ be the edge weighting of $G$ defined by
 $\blambda(e)=(\bbeta(e)-1)/(\bbeta(e)+1)$.  
Then,
for any set $S\subseteq V$,
\begin{equation}\label{eq:worm}
\Ex_{\piI}\left[\,\prod_{v\in S}\sigma(v)\right]=\frac{\ZE{S}}{\ZE\emptyset}.
\end{equation} 
\end{lemma}

\begin{proof}
Observe that 
$$
\wtI(\sigma)=\prod_{e=\{u,v\}\in E}\frac{\bbeta(e)+1}2\left[1+\frac{\bbeta(e)-1}{\bbeta(e)+1}\,\sigma(u)\sigma(v)\right].
$$
since the factor corresponding to $e=\{u,v\}$ contributes $\bbeta(e)$ if $\sigma(u)=\sigma(v)$ and contributes~1 otherwise.  Thus,
setting $\blambda$ as in the statement of the lemma, 
\begin{align}
\ZI=\sum_{\sigma}\wtI(\sigma) &=\prod_{e\in E}\frac{\bbeta(e)+1}2\,\sum_\sigma\prod_{e=\{u,v\}\in E}
  \big[1+\blambda(e)\sigma(u)\sigma(v)\big]\notag\\
&=\prod_{e\in E}\frac{\bbeta(e)+1}2\,\sum_\sigma\sum_{A\subseteq E}\,\prod_{e=\{u,v\}\in A}\blambda(e)\sigma(u)\sigma(v)\notag\\
&=2^n\prod_{e\in E}\frac{\bbeta(e)+1}2 
\sum_{A\in\Omega_\emptyset}\,\prod_{e\in A}\blambda(e)\notag\\
&=c\,\ZE\emptyset,\label{eq:vdW}
\end{align}
where $c=2^n\prod_{e\in E}[(\bbeta(e)+1)/2]$, and $\sigma$ ranges over configurations $V\to\{-1,+1\}$.  
The third equality is explained as follows.  If $(V,A)$ contains an odd degree vertex $u$, 
then $\sigma(u)$ appears to an odd power in the term corresponding to~$A$;  
the term is then annihilated by the summation over~$\sigma$.

Arguing similarly,
\begin{align}
\sum_\sigma\wtI(\sigma)\prod_{w\in S}\sigma(w)&=\prod_{e\in E}\frac{\bbeta(e)+1}2\,\sum_\sigma\prod_{e=\{u,v\}\in E}
  \big[1+\blambda(e)\sigma(u)\sigma(v)\big]\prod_{w\in S}\sigma(w)\notag\\
&=\prod_{e\in E}\frac{\bbeta(e)+1}2\,\sum_\sigma\sum_{A\subseteq E}\,
  \prod_{e=\{u,v\}\in A}\blambda(e)\sigma(u)\sigma(v)\prod_{w\in S}\sigma(w)\notag\\
&=2^n\prod_{e\in E}\frac{\bbeta(e)+1}2 
\sum_{A\in\Omega_S}\,\prod_{e\in A}\blambda(e)\notag\\
&=c\,\ZE{S}\label{eq:vdWS}.
\end{align}
The identity in the statement of the lemma is obtained by dividing \eqref{eq:vdWS} by \eqref{eq:vdW}.
\end{proof}

We remark that the interesting case of the lemma is when $|S|$ is even.
If $|S|$ is odd, then both sides of identity~\eqref{eq:worm} are zero.  
Lemma~\ref{lem:worm} provides a way to approximate the correlation $\Ex[\sigma(s)\sigma(t)]$ 
in the Ising model
by 
estimating the ratio of two partition functions in the even-subgraphs model.
At first sight it might seem that existing Markov chain Monte Carlo approaches might be up to this task.    
One such Markov chain is the so-called ``worm process''. The state space of this chain is defined as follows.

\begin{definition}\label{def:Omega}
Let $\Omega=\bigcup_{S\subseteq V:|S|\leq 2}\Omega_S=\bigcup_{S\subseteq V:|S|\in\{0,2\}}\Omega_S$.  
\end{definition}

The ``worm process'' is  a Markov chain on $\Omega$ whose stationary distribution assigns probability proportional 
to $\blambda(A)=\prod_{e\in A}\blambda(e)$ to each configuration $A\in\Omega$.  
A~transition of the worm process simply flips a single edge of the 
graph from 
being in the configuration~$A$  to being out of~$A$   or vice versa.  
 Thus, as transitions   occur,
the two odd degree vertices move in random paths along the edges of~$G$, occasionally becoming adjacent and disappearing.
   
The worm process is rapidly mixing, as was shown by Collevecchio, Garoni, Hyndman and Tokarev~\cite[Theorem 1.3]{CGHT}.
In principle, to estimate the ratio appearing in the right-hand side
of equation~\eqref{eq:worm}
with  $S=\{s,t\}$, we could just 
run the worm process and
observe the relative time that the process spends in states in $\Omega_{\{s,t\}}$
compared with the time that it spends in states in $\Omega_\emptyset$,
However, if the spins at $s$ and $t$ are only weakly correlated, then the ratio $\ZE{\{s,t\}}/\ZE\emptyset$ will
be small, and the process will spend a small (possibly exponentially small) proportion of time in $\Omega_{\{s,t\}}$.

Following~\cite{JSV},
we modify the worm process by artificially weighting
configurations so that  each subset in the partition $\{\Omega_S:|S|\leq2\}$ of $\Omega$ has roughly equal weight in the 
stationary distribution.   
We will give the details of the modified process in Section~\ref{sec:weightedworm}.
First, we need to 
define the Random Cluster model~\cite{GrimmettBook} 
(which, in the special case we consider, is also equivalent to the Ising model) and 
use the Random Cluster model to prove a lemma (Lemma~\ref{lem:weightcompare} below),
 which will help with the analysis of the weighted worm process.
 
An instance of the Random cluster model is
a graph $G=(V,E)$ with an edge weighting  $\bp:E\to\ourrange\cap(0,1)$.  
A configuration of this model is a subset $A\subseteq E$.
The weight of configuration~$A$ is 
$$\wtrc(A)=\prod_{e\in A}\bp(e)\prod_{e\in E\setminus A}(1-\bp(e))\,2^{\kappa(A)},$$ 
where $\kappa(A)$ is the number of connected components in the graph $(V,A)$.
There is an associated partition function $\Zrc=\sum_{A\subseteq E}\wtrc(A)$, 
but we are more concerned with the probability distribution on configurations given by 
$\pirc=\wtrc(A)/\Zrc$ for all $A\subseteq E$.  
Following Fortuin and Kasteleyn~\cite{FK72}, Edwards and Sokal~\cite{EdwardsSokal}
showed that there is a simple coupling between the distributions $\piI$ and $\pirc$ given by the 
following trial.

\begin{definition} (Edwards-Sokal Distribution)\label{def:SW}
Given a graph $G=(V,E)$ with an edge weighting 
$\bp:E\to\ourrange\cap(0,1)$, let $\calD_{G,\bp}$ be the following distribution on pairs $(A,\sigma)$.
\begin{enumerate}
\item Select $A\subseteq E$ according to the distribution $\pirc$.
\item Independently and uniformly, for each connected component  of $(V,A)$, choose a 
spin   from $\{-1,+1\}$ and assign that spin to all vertices in the  connected component.
Let $\sigma:V\to\{-1,+1\}$ be the resulting 
   spin configuration.
\end{enumerate}
\end{definition}

The following lemma shows that
the output of the Edwards-Sokal coupling is a sample from~$\piI$.

\begin{lemma}(Edwards and Sokal \cite{EdwardsSokal})\label{lem:SW}
Let $G=(V,E)$ be a  graph with edge weighting  $\bbeta:E\to \ourrange_{>1}$.
Let $\bp$ be the edge weighting of $G$ defined by
$\bp(e)=1-1/\bbeta(e)$.  
Let $(A,\sigma)$ be drawn from the Edwards-Sokal distribution $\calD_{G,\bp}$.
Then the distribution of~$\sigma$ is~$\piI$.
\end{lemma}

We say that an event $\calE\subseteq2^E$ is monotonically increasing 
if, for all $A\subset A'\subseteq E$, we have $A\in \calE$ implies $A'\in \calE$.
In the random cluster model as defined here, monotonically increasing events are positively correlated.

\begin{lemma}\label{lem:FKG}
Suppose that events $\calE_1,\calE_2\subseteq 2^E$ are monotonically increasing. 
Then 
$$\Pr_{\pirc}(\calE_1\wedge\calE_2)\geq \Pr_{\pirc}(\calE_1)\Pr_{\pirc}(\calE_2).$$
\end{lemma}

\begin{proof}
This inequality is stated as Part~(b) of Theorem~(3.8) of \cite{GrimmettBook},
for the situation where $\bp(e)$ is the same for all edges~$e$.  However the proof  
is essentially the same when $\bp(e)$ varies with~$e$.
The main step, in order to apply the FKG inequality, is to prove the well-known fact that the 
distribution~$\pirc$ satisfies the FKG lattice condition,  
which says that, for any sets $A_1,A_2 \subseteq E$,
$$\Pr_{\pirc}(A_1 \cup A_2) \Pr_{\pirc}(A_1 \cap A_2) \geq \Pr_{\pirc}(A_1) \Pr_{\pirc}(A_2).$$
To see this, recall the definition of $\pirc$. The denominators cancel, so the FKG lattice condition 
is equivalent to 
$$\wtrc(A_1 \cup A_2) \wtrc(A_1 \cap A_2) \geq \wtrc(A_1)  \wtrc(A_2).$$
Recalling the definition of $\wtrc$, note that, for any edge~$e$, the
quantities $\bp(e)$ and $1-\bp(e)$ occur the same number of times on the left-hand-side and right-hand-side.
Thus, the FKG lattice condition 
is equivalent to
$ 2^{\kappa(A_1 \cup A_2)} 2^{\kappa(A_1 \cap A_2)} \geq 2^{\kappa(A_1)}  2^{\kappa(A_2)}$.
The proof in~\cite{GrimmettBook} now applies without any further changes. 
\end{proof}
 
The following lemma will be used in the analysis of the weighted worm process.

\begin{lemma}\label{lem:weightcompare}
Let $G=(V,E)$ be a graph with edge weighting $\blambda:E \to \ourrange\cap (0,1)$.
Suppose $S,S'\subseteq V$ are subsets of $V$ of even cardinality, 
and assume that it is not the case that $\emptyset\subset S'\subset S$. 
Then 
$$
\frac{\ZE\emptyset}{\ZE{S}}
\leq 
\frac{\ZE\emptyset}{\ZE{S'}}
\times
\frac{\ZE\emptyset}{\ZE{S\oplus S'}}.
$$
\end{lemma}

\begin{proof}
Fix $G=(V,E)$, $\lambda$, $S$ and $S'$ as in the statement of the lemma.
Let $\bbeta$ be the edge weighting of $G$ defined by $\bbeta(e)=(1+\blambda(e))/(1-\blambda(e))$.
Taking reciprocals, the inequality in the statement of the lemma is equivalent  by Lemma~\ref{lem:worm} to
\begin{equation}\label{eq:spincompare}
\Ex_{\piI}\left[\,\prod_{v\in S}\sigma(v)\right]\geq\Ex_{\piI}\left[\,\prod_{v\in S'}\sigma(v)\right]\times\Ex_{\piI}\left[\,\prod_{v\in S\oplus S'}\sigma(v)\right].
\end{equation} 

Now let $\bp$ be the edge weighting defined by 
$\bp(e)=1-1/\bbeta(e)$. 
Let $(A,\sigma)$ be drawn from the Edwards-Sokal distribution $\calD_{G,\bp}$.
For any subset $T$ of~$V$, let 
``$T$ is connected'' be a shorthand for the event ``$T$ is contained within a single connected component of $(V,A)$''.
Let $Y_T$ be the random variable $Y_T=\prod_{v\in T}\sigma(v)$. Then 
\begin{align*}
\Ex_{\DSW}[Y_T]=\Pr_{\DSW}(\text{$T$ is connected}) & \Ex_{\DSW}[Y_T\mid\text{$T$ is connected}]\\
  &+\Pr_{\DSW}(\text{$\neg$ $T$ is connected})\Ex_{\DSW}[Y_T\mid\text{$\neg$ $T$ is connected}].
\end{align*}

The definition of  $\calD_{G,\bp}$ (Definition~\ref{def:SW})
ensures that, for any set~$T$ with even cardinality,
$\Ex_{\DSW}[Y_T\mid\text{$T$ is connected}]=1$ and $\Ex_{\DSW}[Y_T\mid\text{$\neg$ $T$ is connected}]=0$.
Hence, $\Ex_{\DSW}[Y_T]=\Pr_{\DSW}(\text{$T$ is connected})$.
Using Lemma~\ref{lem:SW},
$$\Ex_{\piI}[Y_T] =  \Ex_{\DSW}[Y_T]=
\Pr_{\DSW}(\text{$T$ is connected})=
\Pr_{\pirc}(\text{$T$ is connected}).$$
Plugging this into~\eqref{eq:spincompare} with $T=S$, $T=S'$ and $T=S\oplus S'$,
we find that~\eqref{eq:spincompare} is equivalent  to the following inequality.
\begin{equation}\label{eq:almostFKG}
\Pr_{\pirc}(\text{$S$ is connected})\geq\Pr_{\pirc}(\text{$S'$ is connected})\times\Pr_{\pirc}(\text{$S\oplus S'$ is connected}).
\end{equation}

By considering the   possible intersections of~$S$ and~$S'$, 
recalling from the statement of the lemma 
that it is not the case that $\emptyset\subset S'\subset S$,
it is 
easy to see that
\begin{equation}\label{eq:easyineq}
 \Pr_{\pirc}(\text{$S$ is connected})\geq 
 \Pr_{\pirc}(\text{$S'$ is connected } \wedge \text{ $S\oplus S'$ is connected}).
\end{equation}
Now observe that ``$S'$ is connected'' and ``$S\oplus S'$ is connected'' are both monotonically
increasing events, and hence \eqref{eq:almostFKG} follows from \eqref{eq:easyineq} by Lemma~\ref{lem:FKG}. 
\end{proof}

\section{The weighted worm process}\label{sec:weightedworm}

Consider a graph $G=(V,E)$ with an edge weighting
$\blambda:E \to \ourrange\cap (0,1)$.

\begin{definition} \label{def:LamA}
A \emph{subset weighting} of~$G$
 is a function $w$ that assigns a weight $w_S\in \ourrange_{>0}$
to each subset $S$ of~$V(G)$
with $|S|\in\{0,2\}$.
We  refer to the pair $(\blambda,w)$ as a \emph{weighting} of~$G$.
Given a subset  $A\subseteq E(G)$, 
there is a unique $S(A)\subseteq V(G)$ such that $A\in \Omega_{S(A)}$.
If $|S(A)|\leq2$ we 
define $\Lambda(A) = \blambda(A) w_{S(A)}$. 
The partition function that we study is 
$$\hatZS{S} = \sum_{A\in\Omega_S}  \Lambda(A) = 
\sum_{A \in \Omega_S} \blambda(A) w_S
= w_S \ZE{S}
.$$
We also define $\hatZ = \sum_{S\subseteq V; |S| \leq 2} \hatZS{S}$.
\end{definition}

Later, we shall need to extend the above definition to subsets $S\subseteq V(G)$ 
with $|S|\leq4$ in the obvious way.
 
Recall from Definition~\ref{def:Omega} that
$\Omega=\bigcup_{S\subseteq V:|S|\leq2}\Omega_S$.    
The  weighted worm process   
is a Markov chain with state space~$\Omega$. 
The transitions of the process are given in Figure~\ref{fig:transition}.
\begin{figure}
\begin{algorithmic}
\STATE{ (* One transition from state $A\in \Omega$ *)}
\STATE{Choose the type of transition $T$ uniformly at random from $\{\text{``self-loop''},\text{``move''}\}$}\\
\IF {$T = \text{``self-loop''}$} 
      \STATE {the next state is $A$} 
\ELSE
{ 
\STATE {Choose an edge $e\in E$ uniformly at random}
\IF {$A\oplus\{e\}\in\Omega$} \STATE {$A'\leftarrow A\oplus\{e\}$} \ELSE \STATE {$A'\leftarrow A$} \ENDIF
\STATE {With probability $\min\{\Lambda(A')/\Lambda(A),1\}$ the next state is $A'$, otherwise $A$}
}\ENDIF
\end{algorithmic}
\caption{One transition of the weighted worm process for graph $G=(V,E)$
with weighting $(\blambda,w)$, starting at state $A\in\Omega$,
where 
$\Lambda(A) = \blambda(A) w_{S(A)}$. } \label{fig:transition}
\end{figure}
It is easy to see from the definition of the transitions that the weighted worm process is ergodic
and time-reversible
and that the stationary probability of
a configuration $A\in \Omega$ is $\pi(A) 
=\Lambda(A) / \hatZ$.

Given a subset $S$ of $V(G)$ with $|S|\leq 2$,
the probability of $\Omega_S$ in the stationary distribution of the weighted worm process
is 
$$\sum_{A\in \Omega_S} \pi(A)= \frac{w_S \ZE{S}}{\hatZ}
= \frac{w_S \ZE{S}}{\sum_{S'} w_{S'} \ZE{S'}},$$
where the sum is over all subsets $S' \subseteq V(G)$ with $|S'| \leq 2$.

Thus, we will be most interested in the weighted worm process
when the weighting satisfies 
$w_S =\ZE\emptyset/\ZE{S}$
so that all subsets~$S$ have equal weight.
We show in Section~\ref{sec:learn} how 
to ``learn'' such a weighting by running the process multiple times.
First, we consider the mixing rate of the process itself.

\subsection{Rapid mixing of the weighted worm process}\label{sec:CPC2}
 
In broad outline, the proof of rapid mixing follows existing work~\cite{JS93,CGHT}, 
but is complicated by the need to deal with the subset weightings.

We use $\calW(G)$ to denote the set of  weightings $(\blambda,w)$
where 
$\blambda:E \to \ourrange\cap (0,1)$ is an edge weighting of $G$
and $w$ is a subset weighting of~$G$
 satisfying  
  \begin{align}\nonumber
w_S=1, \quad &   \text{if $|S|=0$,} \\   \label{eq:weights}
w_S=0, \quad & \text{if $|S|=1$, and} \\ \nonumber
\frac{1}{2} \leq \frac{\hatZS{S}}{\hatZS\emptyset} \leq 2, \quad & \text{if $|S|=2$.} 
 \end{align}

The purpose of this section is to prove 
that the weighted worm process is rapidly mixing if 
$(\blambda,w)\in \calW(G)$ (see Lemma~\ref{lem:mix} below).

In order to prove rapid mixing, given a weighting $(\blambda,w)$
of~$G$
it will be useful to extend the subset weighting $w$  by 
defining $w_S = \ZE\emptyset/\ZE{S}$
for every $S$ with $|S|=4$.
The extended weighting will be used in the proof, but not in the Markov chain.
The following lemma
will be used in the proof of rapid mixing. 

\begin{lemma}
\label{lem:weights}
If $(\blambda,w) \in \calW(G)$ then, 
for every subset~$S$ of $V(G)$ with $|S|\in \{0,4\}$
we have $w_S =   {\ZE\emptyset}/{\ZE{S}}$.
For every size-$2$ subset $S$ of $V(G)$ we have
$$  \frac{\ZE\emptyset}{2 \ZE{S}} \leq 
w_S
\leq 
 \frac{2 \ZE\emptyset}{\ZE{S}}.$$
\end{lemma}
 \begin{proof}
 The lemma 
 follows trivially from the definition of $w_S$ if $|S|=0$ or $|S|=4$, so suppose that $|S|=2$. 
 From~\eqref{eq:weights} 
 and the definitions of $\ZE{S}$ and $\hatZS{S}$
 we have
 $$
 \frac{1}{2} \leq 
 \frac
{ w_S  \ZE{S}  }
{  \ZE\emptyset }  \leq 2, $$
as required.
 \end{proof}

In order to bound the mixing time of the weighted worm process we use the canonical path method
or, more precisely, a well-known generalisation of the method that replaces paths by flows.
We briefly describe the method, using notation that is slightly more general than that of the weighted worm process.  
Consider a Markov chain~$\calM$  on a state space $\Omegastar$  with transition 
matrix $P$ and stationary distribution~$\pistar$.
 A \emph{path} from a state~$I\in \Omegastar$ to
 a state~$F\in \Omegastar$ is  a sequence $I=T_0,\ldots,T_k=F$ of   states, all of which are distinct except possibly~$I$ and~$F$,
 such that, for each $i\in \{0,\ldots,k-1\}$, $P_{T_i,T_{i+1}} >0$.
 A \emph{flow} $f_{I,F}$ is
a distribution whose support is the set of paths from~$I$ to~$F$
which is 
 normalised so that $\sum_pf_{I,F}(p)=\pistar(I)\pistar(F)$.  
 Typically, when we refer to a flow $f_{I,F}$, we refer to~$I$ as the ``initial state'' and
 to~$F$ as the ``final state''.
The collection of all flows is $\calF=\{f_{I,F}:I,F\in\Omegastar\}$.  
The \emph{congestion} of this collection of flows is 
\begin{equation}\notag 
\rho(\calF)=\max_{(T,T')}\left\{\frac1{\pi(T)P(T,T')}\sum_{I,F\in\Omegastar}\,\,\sum_{p=I,\ldots ,T,T',\ldots ,F}
  f_{I,F}(p)\> |p|\right\},
\end{equation} 
where the maximisation is over all transitions $(T,T')$ with $P(T,T')>0$, 
the second sum is over all  paths~$p$ from $I$ to $F$   that use transition $(T,T')$,
and $|p|$  denotes the length of  path~$p$.

The mixing time $\tmix{T}(\delta)$ of~$\calM$, when starting from state~$T$, is defined
to be the minimum time~$t$ such that the total variation distance between the $t$-step distribution
$P^t(T,\cdot)$ and the stationary distribution~$\pistar$
of~$\calM$
 is at most~$\delta$.
The existence of a collection of flows with small congestion implies rapid mixing.
The following lemma  is due to Sinclair~\cite{Sin92}, building on work of Diaconis and Stroock~\cite{DS91}.
The explicit statement that we use is taken from \cite[Lemma 2.2]{JSV}
\begin{lemma}\label{lem:mixvscongestion}
Let $\calM$ be an ergodic
   time-reversible Markov chain with state space $\Omegastar$ and stationary distribution~$\pistar$ 
  whose self-loop probabilities satisfy $P(T,T)\geq 1/2$ for all states $T$.
  Suppose that $\calM$ supports a collection $\calF$ of flows.
  Given any  state $T_0\in \Omegastar$,
   $$\tmix{T_0}(\delta) \leq \rho(\calF) \left(\ln \left(\frac{1}{\pistar(T_0)}\right)+\ln\left(\frac{1}{\delta}\right)\right).$$
\end{lemma}

A standard method for defining a collection of flows
is to partition the state space~$\Omegastar$ into two parts~$\Omegastar_1$ and~$\Omegastar_2$,
define canonical paths from every state~$I\in \Omegastar_1$ to every state~$F\in \Omegastar_2$,
and then use an idea similar to Valiant's randomised routing~\cite{Val82}
to obtain a collection of flows.
 Thus, for each pair or initial and final states $(I,F)\in\Omegastar_1\times\Omegastar_2$ 
we specify a path $\cp IF$   from $I$ to~$F$. 
The collection of all such canonical paths is $\Gamma=\{\cp IF:(I,F)\in\Omegastar_1\times\Omegastar_2\}$.  
For each possible transition $(T,T')$ of the Markov chain, denote by 
$$\cps(T,T')=\big\{(I,F)\in\Omegastar_1\times\Omegastar_2:\text{$\cp IF$ uses the transition $(T,T')$}\big\}$$ 
the set of canonical paths using $(T,T')$.
The \emph{congestion} of~$\Gamma$ is then given by 
\begin{equation}\label{eq:pathcongestion}
\rho(\Gamma)=\max_{(T,T')}\left\{\frac1{\pistar(T)P(T,T')}\sum_{(I,F)\in\cps(T,T')}\pistar(I)\pistar(F)\>|\gamma(I,F)|\right\}.
\end{equation} 

The next step is to use the canonical paths in $\Gamma$
to induce a collection~$\calF$ of flows, via randomised routing:
If $I$ and $F$ are in $\Omegastar_1$ then
the flow $f_{I,F}$ is constructed by choosing  
intermediate states $T\in \Omegastar_2$ and routing flow via the 
path $\gamma(I,T)$ followed by the reversal of the path $\gamma(F,T)$.
Similarly, flow from $\Omegastar_2$ to~$\Omegastar_2$ is routed via a random intermediate
state in $\Omegastar_1$.
The following lemma 
shows that if the congestion~$\rho(\Gamma)$ is low then 
 the resulting collection~$\calF$ also has low congestion.
 The lemma is a direct translation of
 Lemma~4.4 of~\cite{JSV} into the more general language of this section.  A similar lemma was used earlier by Schweinsberg~\cite{Schw02}.

 \begin{lemma}\label{lem:pathstoflows}
Given a partition $\{\Omegastar_1,\Omegastar_2\}$ of the state space $\Omegastar$ of a time-reversible Markov chain,
and a collection $\Gamma$ of canonical paths from $\Omegastar_1$ to~$\Omegastar_2$ with congestion $\rho(\Gamma)$,  
there exists a collection of flows~$\calF$ with congestion
$$
\rho(\calF)\leq \left(2+4\left(\frac{\pistar(\Omegastar_1)}{\pistar(\Omegastar_2)}+\frac{\pistar(\Omegastar_2)}{\pistar(\Omegastar_1)}\right)\right)\rho(\Gamma).
$$
\end{lemma}

A bound on the mixing time of the Markov chain~$\calM$ 
can be derived by 
constructing low-congestion canonical paths from $\Omegastar_1$ to $\Omegastar_2$,
using Lemma~\ref{lem:pathstoflows} to derive a collection of flows
with low congestion, and then applying  Lemma~\ref{lem:mixvscongestion}.
We next apply these methods to bound the mixing time of the weighted worm process. 

\begin{lemma}
\label{lem:mix}
Suppose that $G=(V,E)$ is a connected graph with $n$ vertices and $m$ edges and
$(\blambda,w)\in \calW(G)$.
Let $\lambdamin=\min_{e\in E}\lambda(e)$. 
 Then the  weighted worm process, started in the empty configuration on~$G$,
has mixing time 
$\tmix{\emptyset}(\delta) =
O(\lambdamin^{-2}n^4m^2)\left( O(m)+\ln\left(\frac{1}{\delta}\right)\right) $.
\end{lemma}
 
\begin{proof}

As we observed earlier, the weighted worm process is
a time-reversible Markov chain with state space $\Omega=\bigcup_{S\subseteq V:|S|\leq2}\Omega_S$.
Our goal will be to apply Lemma~\ref{lem:pathstoflows}.
To this end,  let $\Omegabar=\Omega\setminus\Omega_\emptyset=\bigcup_{S\subseteq V:|S|=2} \Omega_S$.   
We will define a collection~$\Gamma$ of canonical paths from~$\Omegabar$ to~$\Omega_\emptyset$.
We will bound the congestion $\rho(\Gamma)$ and use Lemma~\ref{lem:pathstoflows} and Lemma~\ref{lem:mixvscongestion}
to bound the mixing time.

We start by constructing a canonical path  from 
any configuration  $I\in\Omegabar$  to any configuration $F\in\Omega_\emptyset$.   
Let $a$ and $b$ be the two odd-degree vertices in $I$.
The vertices of the graph $(V,I\oplus F)$  all have even degree, except for $a$ and~$b$.  
Choose a canonical partition of $I\oplus F$ 
into a path~$\Pi$  from $a$ to~$b$, and a number of cycles $C_1,C_2,\ldots,C_k$;  also choose a distinguished
end vertex for $\Pi$ and a distinguished vertex and orientation for each cycle. To \emph{unwind} a path or cycle, 
start at the distinguished vertex and travel along the path or around the oriented cycle flipping all edges 
along the way.  The act of \emph{flipping} changes the status of an edge from absent to present or vice versa.
The canonical path $\cp I F$ is obtained 
by unwinding first the path~$\Pi$ and then the cycles $C_1,C_2,\ldots,C_k$, in order.
Note that all of the flips are transitions of the Markov chain corresponding to the weighted worm process.

Fix any transition $(T, T')$
that can be made by the Markov chain, and let 
$$\cps(T, T')=\{(I,F)\in \Omegabar \times \Omega_\emptyset \colon \text{$\cp IF$  uses the transition $(T, T')$}\}$$  
be the set of canonical paths using this transition.  
Our goal is to bound the congestion through the transition $(T, T')$.

Denote by $\Omegahat$ the extended state space $\Omegahat=\bigcup_{S\subseteq V:|S|\leq4}$.
Consider the function $\enc:\cps(T, T')\to\Omegahat$ defined by $\enc(I,F)=I\oplus F\oplus T$.
(Note that the range of $\enc$ is contained in $\Omegahat$.)
We claim that $\enc$ is injective.  To see this, suppose   $(I,F)\in\cps(T, T')$ and let $X=\enc(I,F)$.
Since $I\oplus F=T\oplus X$, we can  
use the fixed configuration~$T$ from the transition
and the known value~$X$ to
recover $I\oplus F$ and hence the path $\Pi$ and the cycles 
$C_1,C_2,\ldots,C_k$.  The set $T\oplus T'$ contains a single edge~$e$, which tells which of 
$\Pi,C_1,C_2,\ldots,C_k$ is 
having its edges flipped by the particular transition $(T, T')$.   Using this information, we can apportion the 
edges in $\Pi\cup C_1\cup C_2\cup\cdots\cup C_k=I\oplus F$ between $I$ and~$F$.  
(Each edge is either in $I$ or $F$ but not both.)  
Say that $I\oplus F$ is the disjoint union of $I'$ and~$F'$,
with $I'\subseteq I$ and $F'\subseteq F$. 
We can then recover $I$ and $F$ themselves  using the equalities $I=I'\cup(I\cap F)=I'\cup(T\cap X)$\footnote{
To see that $I\cap F = T\cap X$ consider some $e\in I \cap F$. Note from the definition of $\cp I F$ that
$e$ is in every configuration along way from~$I$ to $F$. Hence $e$ is in $T$. By the definition of~$X$, $e$ is also in $X$. The other direction is similar.} 
and 
$F=F'\cup(I\cap F)=F'\cup(T\cap X)$.   
Thus, we have shown that $\enc$ is injective.

We now proceed to bound the congestion through the transition $(T, T')$.
Note that $\blambda(I)\blambda(F)=\blambda(T)\blambda(X)$, which is exactly what we would need for the analysis of
the unweighted case, i.e., when $w_S=1$ for all $S$.  To analyse the weighted case we need to relate 
$\Lambda(I)\Lambda(F)$ to $\Lambda(T)\Lambda(X)$.  There are three cases, depending on where the transition $(T, T')$
occurs on the canonical path from $I$ to $F$.  
\begin{itemize}

\item The transition is the first one of all.  Then $T=I$ and $X=F$, and so
$\Lambda(I)\Lambda(F)=\Lambda(T)\Lambda(X)$.

\item The transition is on the unwinding of the path $\Pi$.  Then $T\in\Omega_{\{b,c\}}$, where $c$ is a vertex on the
path~$\Pi$, from which it follows that $X\in\Omega_{\{a,c\}}$. We will show
$$ 
\Lambda(I)\Lambda(F)= w_{\{a,b\}}\blambda(I)w_\emptyset\blambda(F)\leq \KK w_{\{b,c\}}\blambda(T)w_{\{a,c\}}\blambda(X)
=\KK\Lambda(T)\Lambda(X).
$$
To establish the inequality
recall that $\blambda(I)\blambda(F)=\blambda(T)\blambda(X)$
so, cancelling these out, and noting that $w_\emptyset=1$, it suffices to show 
$ w_{\{a,b\}}   \leq \KK w_{\{b,c\}} w_{\{a,c\}}$.
Using Lemma~\ref{lem:weights},
it suffices to show
$$ 
 \frac{2 \ZE\emptyset}{ \ZE{\{a,b\}}}
\leq \KK 
\frac{\ZE\emptyset}{2 \ZE{\{b,c\}}}\>
\frac{\ZE\emptyset}{2 \ZE{\{a,c\}}},$$
 which follows from Lemma~\ref{lem:weightcompare}
taking $S=\{a,b\}$ and $S' = \{b,c\}$.

\item The transition is the first one in the unwinding of a cycle.  
Then $T\in\Omega_\emptyset$ and
and $X\in\Omega_{\{a,b\}}$, and hence
$$ 
\Lambda(I)\Lambda(F)=w_{\{a,b\}}\blambda(I)w_\emptyset\blambda(F)=w_\emptyset\blambda(T)w_{\{a,b\}}\blambda(X)=\Lambda(T)\Lambda(X).
$$
 
\item The transition arises during the unwinding of a cycle but is not the first such transition. 
Then $T\in\Omega_{\{c,d\}}$ for vertices $c$ and~$d$ on the cycle,
and $X\in\Omega_{\{a,b,c,d\}}$.  
We will show 
$$ 
\Lambda(I)\Lambda(F)= w_{\{a,b\}}\blambda(I)w_\emptyset\blambda(F)\leq \KK w_{\{c,d\}}\blambda(T)w_{\{a,b,c,d\}}\blambda(X)
=\KK\Lambda(T)\Lambda(X).
$$
As in the second case, it suffices to show 
$ w_{\{a,b\}}   \leq \KK w_{\{c,d\}} w_{\{a,b,c,d\}}$.
Using Lemma~\ref{lem:weights},
it suffices to show
$$ 
 \frac{2 \ZE\emptyset}{ \ZE{\{a,b\}}}
\leq \KK 
\frac{\ZE\emptyset}{2 \ZE{\{c,d\}}}\>
\frac{\ZE\emptyset}{\ZE{\{a,b,c,d\}}},$$
 which follows from Lemma~\ref{lem:weightcompare}
 (with a factor of~$2$ to spare)
taking $S=\{a,b\}$ and $S' = \{c,d\}$. 

\end{itemize}

Note that in all instances, $\Lambda(I)\Lambda(F)\leq \KK \Lambda(T)\Lambda(X)$. Given a set $\Psi \subseteq \Omegahat$,
we use $\Lambda(\Psi)$ to denote $\sum_{C\in \Psi} \Lambda(C)$. 
The probability of a configuration~$C\in \Omega$ in the 
stationary distribution of the weighted worm process is then
$\pi(C) = \Lambda(C)/\Lambda(\Omega)$. 
 We can then  bound the congestion through transition $(T, T')$ arising from the canonical paths $\Gamma$ as follows.
\begin{align}
\sum_{(I,F)\in\cps(T, T')}\pi(I)\pi(F)&=\frac1{\Lambda(\Omega)^2}\sum_{(I,F)\in\cps(T, T')}\Lambda(I)\Lambda(F)\notag\\
&\leq\frac\KK{\Lambda(\Omega)^2}\sum_{(I,F)\in\cps(T, T')}\Lambda(T)\Lambda(\enc(I,F))\notag\\
&\leq \frac\KK{\Lambda(\Omega)^2}\sum_{X\in\Omegahat}\Lambda(T)\Lambda(X)\notag\\
&=\KK\times\frac{\Lambda(\Omegahat)}{\Lambda(\Omega)}\times\frac{\Lambda(T)}{\Lambda(\Omega)}\notag\\
&= O(n^2)\pi(T)\label{eq:flowbd}.
\end{align}
The second inequality uses the fact that $\enc$ is injective.
The final equality follows from  
the observation that 
$$
\Lambda(\Omega)=\sum_{S:|S|\in\{0,2\}}\hatZS{S}
\quad\text{and}\quad
\Lambda(\Omegahat)=\sum_{S:|S|\in\{0,2,4\}}\hatZS{S}.
$$
The first sum has $O(n^2)$ terms and the second $O(n^4)$.
Thus,  
$$
\frac{\Lambda(\Omegahat)}{\Lambda(\Omega)} \leq
O(n^2) \frac{\max_{S:|S|\in\{0,2,4\}}\hatZS{S}}
{\min_{S:|S|\in\{0,2\}}\hatZS{S}}
=  
O(n^2) \frac{\max_{S:|S|\in\{0,2,4\}}   w_S \ZE{S} }
{\min_{S:|S|\in\{0,2\}} w_S \ZE{S}}.
$$
By Lemma~\ref{lem:weights},
this is at most 
$$
O(n^2) \frac{2    \ZE{\emptyset} }
{ \tfrac12   \ZE{\emptyset}} = O(n^2),
$$
so the final equality holds.

By establishing~\eqref{eq:flowbd},
we have done most of the work required to estimate the congestion $\rho(\Gamma)$ in~\eqref{eq:pathcongestion}.  
Since  the paths have length at most~$m$, the only remaining task is to lower bound the
transition probability $P(T,T')$.  
Let $e=\{u,v\}\in E$ be any edge, and let $S\subseteq V$ be any subset of 
vertices of even cardinality.  There is a bijection between $\Omega_S$ and $\Omega_{S\oplus\{u,v\}}$ 
obtained by flipping the edge~$e$.  Since this operation changes only a single edge, we see that
$$
\lambda(e)\ZE{S}\leq\ZE{S\oplus\{u,v\}}\leq\lambda(e)^{-1}\ZE{S}.
$$
Then, from Lemma~\ref{lem:weights},
$$
\frac{\lambdamin}4\leq\frac{\ZE{S}}{4\ZE{S\oplus\{u,v\}}}
\leq\frac{w_{S\oplus\{u,v\}}}{w_S}
\leq \frac{4\ZE{S}}{\ZE{S\oplus\{u,v\}}}\leq\frac4{\lambdamin}.
$$
Since $T$ and $T'$ differ by a single edge, this implies
$$
\frac{\lambdamin^2}4\leq\frac{\Lambda(T')}{\Lambda(T)}\leq \frac4{\lambdamin^2}.
$$

Going back to the definition of the weighted worm process in Figure~\ref{fig:transition},
it follows that
$P(T,T') \geq \tfrac12 \tfrac1m \min\{\Lambda(T')/\Lambda(T),1\}
\geq \lambdamin^2 /(8m)
$.

Now, starting from~\eqref{eq:pathcongestion},
and plugging in the bound  that path-lengths are at most~$m$ and~\eqref{eq:flowbd} and then 
this bound, we get
\begin{align*}
\rho(\Gamma)&=\max_{(T,T')}\left\{\frac1{\pi(T)P(T,T')}\sum_{(I,F)\in\cps(T,T')}\pi(I)\pi(F)\>|\gamma(I,F)|\right\}\\
&\leq \max_{(T,T')}\left\{\frac1{\pi(T)P(T,T')} O(n^2) \pi(T) m\right\} =O(\lambdamin^{-2}n^2m^2)\\
\end{align*}

In order to apply  Lemma~\ref{lem:pathstoflows}
we must find an upper bound for
$\pi(\Omegabar)/\pi(\Omega_\emptyset)$ and
$\pi(\Omega_\emptyset)/\pi(\Omegabar)$.
Using the upper bound in Lemma~\ref{lem:weights},
$$
\frac{\pi(\Omegabar)}{\pi(\Omega_\emptyset)} = 
\frac{\sum_{A \in \Omegabar} \Lambda(A)}
{\sum_{A \in \Omega_\emptyset} \Lambda(A)}
= \frac
{\sum_{S:|S|=2} \hatZS{S}}
{\hatZS{\emptyset}}
= 
\frac
{\sum_{S:|S|=2} w_S \ZE{S}}
{w_\emptyset \ZE{\emptyset}}= O(n^2).
$$
Similarly, $\pi(\Omega_\emptyset)/\pi(\Omegabar) = O(1/n^2) = O(n^2)$. 

Now applying    Lemma~\ref{lem:pathstoflows}, there is a collection of flows~$\calF$ with 
$\rho(\calF)\leq O(n^2)\rho(\Gamma)=O(\lambdamin^{-2}n^4m^2)$.

In order to apply Lemma~\ref{lem:mixvscongestion} starting from state $T_0=\emptyset$ we need an upper bound for
$\ln({1}/{\pi(\emptyset)})$. For this we use
$$
\ln \left(\frac{1}{\pi(\emptyset)}\right) =
\ln \left(\frac{\hatZ}{ \Lambda(\emptyset)} \right) =
\ln(\hatZ).
$$
By the definition of~$\hatZ$ and~\eqref{eq:weights},
$$\ln(\hatZ) \leq \ln( n^2 \hatZS{\emptyset})
= \ln(n^2 \ZE{\emptyset}) \leq \ln(n^2 2^m) = O(m),
$$
where the asymptotic bound uses the fact that $G$ is connected.

Finally, by Lemma~\ref{lem:mixvscongestion},
$$\tmix{\emptyset}(\delta) \leq \rho(\calF) \left( O(m)+\ln\left(\frac{1}{\delta}\right)\right)=
O(\lambdamin^{-2}n^4m^2)\left( O(m)+\ln\left(\frac{1}{\delta}\right)\right).$$
\end{proof}

The following lemma captures how we will use Lemma~\ref{lem:mix}. 
 
\begin{lemma}\label{lem:mcmc}
There is an algorithm that takes as input an $n$-vertex connected graph $G=(V,E)$
with a weighting 
$(\blambda,w)\in \calW(G)$
and a set $S\subseteq V$ with $|S|=2$, also an accuracy parameter $\epsilon\in(0,1)$
and a desired failure probability~$\delta^*$.
With probability at  least $1-\delta^*$, the algorithm produces
as estimate 
$\hatR$ such that
$$e^{-\epsilon} \hatR \leq \frac{\hatZS{\emptyset} }{\hatZS{S} } \leq e^{\epsilon} \hatR.$$
Let $\lambdamin=\min_{e\in E}\lambda(e)$. 
The running time of the algorithm is at most 
a polynomial in~$n$, $1/\lambdamin$, $1/\epsilon$, 
and $\log(1/\delta^*)$.
\end{lemma}

\begin{proof} 
Let $\theta=\epsilon/8$, $\delta=\epsilon/(32 n^2)$ and
 $T= \lceil 
 \ln(6/\delta^*) e^{8 n^2 \delta} 12 n^2/\theta^2
 \rceil$.
 Let $\lambdamin=\min_{e\in E}\lambda(e)$.
 Let $t$ be the
upper bound on the mixing time 
$\tmix{\emptyset}(\delta)$ of the weighted worm process, from Lemma~\ref{lem:mix}.
Given the definition of~$\delta$, $t$ is at most a polynomial in~$n$, $1/\lambdamin$, and $\log(1/\epsilon)$.
For $i\in[T]$,
the algorithm will  
run the weighted worm process for $t$ steps, starting from the empty configuration,
computing
$x_i$, the indicator for the event that the output is in $\Omega_\emptyset$,
Similarly, for $i\in[T]$,
the algorithm will 
run the weighted worm process for $t$ steps, starting from the empty configuration, computing
$y_i$, the indicator for the event that the output is in $\Omega_S$. Let $x= \sum_{i=1}^T x_i$ and $y=\sum_{y=1}^T y_i$.
The output is then $\hatR =  x/y$.
The calculation of errors is standard.  
Let $p_\emptyset = \Lambda(\emptyset) = \hatZS{\emptyset}/\hatZ$.
Since $(\blambda,w)\in \calW(G)$, by the definition of $\calW(G)$,
we have the loose inequality $1/(4n^2) \leq p_\emptyset \leq 4/n^2$.
By the total variation distance guarantee of Lemma~\ref{lem:mix}, the probability $\hat{p}_\emptyset$ that $x_i=1$
satisfies 
$\hat{p}_\emptyset \leq p_\emptyset + \delta = (1+ \delta/p_\emptyset) p_\emptyset \leq e^{\delta/p_\emptyset} p_\emptyset \leq e^{4 n^2 \delta} p_\emptyset$
and
$\hat{p}_\emptyset \geq p_\emptyset-\delta = (1-\delta/p_\emptyset) p_\emptyset \geq e^{-2\delta/p_\emptyset} p_\emptyset \geq e^{-8 n^2\delta} p_\emptyset$.
Then by a Chernoff bound, for any $\theta\in (0,1)$,
$$\Pr(x \geq e^\theta T e^{4 n^2 \delta} p_\emptyset ) \leq
\Pr(x \geq (1+\theta) T \hat{p}_\emptyset) \leq 2 \exp(-\theta^2 \hat{p}_\emptyset T/3) \leq 
2 \exp(-\theta^2T /(e^{8 n^2\delta}12 n^2)). $$
Similarly, 
$$\Pr(x \leq e^{-2\theta} T e^{-8 n^2 \delta} p_\emptyset )
\leq \Pr( x \leq (1- \theta) T \hat{p}_\emptyset) \leq
\exp(-\theta^2 \hat{p}_\emptyset T/2)
\leq 
\exp(-\theta^2  
  T/(e^{8 n^2 \delta}  8 n^2)).$$
Similarly, with $p_S = \Lambda(S) = \hatZS{ S}/\hatZ$,
the probability that $y$  fails to satisfy
$  e^{-2\theta} e^{-8 n^2 \delta} p_S T \leq y \leq 
e^\theta   e^{4 n^2 \delta} p_S T$
is at most 
$ 3 \exp(-\theta^2T /(e^{8 n^2\delta}12 n^2))$.
The accuracy guarantee follows  from the choice of $\theta$ and $\delta$, which ensure that
  $e^{2\theta}   e^{8 n^2 \delta} = e^{\epsilon/2}$.

The failure probability guarantee comes from the fact that
$ 6 \exp(-\theta^2T /(e^{8 n^2\delta}12 n^2)) \leq \delta^*$. 
The worm process is simulated for $t$ steps $O(T)$ times, giving the
running time bound in the statement of the lemma. \end{proof}

\subsection{Learning appropriate weights for the worm process}\label{sec:learn}

Lemma~\ref{lem:mix}
shows that the weighted worm process is rapidly mixing as long as 
the weighting $(\blambda,w)$ is in  $\calW(G)$.
Let $G=(V,E)$ be a connected graph with $|V|=n$ and $|E|=m$.
Let $\blambda:E \to \ourrange\cap (0,1)$ be an edge weighting of~$G$.
 
In this section we show
how to learn a sequence  
$(\lambda^{[0]},w^{[0]}),\ldots,(\lambda^{[t]},w^{[t]})$ of weightings so that each weighting 
$(\lambda^{[i]},w^{[i]})$ is in $\calW(G)$.
The sequence will satisfy 
\begin{equation}\label{eq:lambdai}
\lambda^{[i]}(e) = \max(1/(1+\bd)^i,\lambda(e)),
\end{equation}
so taking $t = \max_{e\in E} \left\lceil \log(1/\lambda(e)) / \log(1+\bd)\right\rceil$, we have
$\lambda^{[t]} = \lambda$.
The results of the section are summarised in Lemma~\ref{lem:learnweights}.

Although the definition of $\lambda^{[i]}$, from Equation~\eqref{eq:lambdai}, is straightforward, the
definition of the subset weighting~$w^{[i]}$ is more complicated. 
In order to conform with the definition~\eqref{eq:weights} of  $\calW(G)$,
we will set $w^{[i]}_\emptyset=1$ for all~$i\in\{0,\ldots,t\}$.
Also, for sets $S$ with $|S|=1$, we set $w^{[i]}_S=0$.
This leaves the definition of $w^{[i]}_S$ where $|S|=2$.
For this, we start by defining the base case, which is $i=0$. Then, we  
show how to learn $w^{[i+1]}$   from $w^{[i]}$ by running the weighted worm process.
As quantified by Lemma~\ref{lem:learnweights},
there is a probability that the process does not converge sufficiently quickly to its stationary
distribution.
Thus, throughout this section we take $\delta$ to be the desired failure probability, from Lemma~\ref{lem:learnweights}.
We will give an algorithm which, with probability at least $1-\delta$, learns the weights.
We start by defining the base case. For every size-$2$ set  $S\subseteq V$,
we  set $w^{[0]}_S=1$.

\begin{observation}
 The weighting
$(\lambda^{[0]},w^{[0]})$ is in $\calW(G)$. 
\end{observation}

\begin{proof}
First note that, for every $e\in E$,
  $\lambda(e) \leq 1$ so  $\lambda^{[0]}(e)=1$.
  
 Consider any $S\subseteq V$ with 
 $|S|=2$ and note that
$$\frac{\hatZp{S}{G;\lambda^{[0]},w^{[0]}}}{\hatZp\emptyset{G;\lambda^{[0]},w^{[0]}}} = 
 \frac {\sum_{A \in \Omega_S}  1 }
{\sum_{A \in \Omega_\emptyset}  1 }.
$$
We will show below that  $|\Omega_{S}| = |\Omega_\emptyset|$.
This ensures 
that $(\lambda^{[0]},w^{[0]})$ satisfies   Equation~\eqref{eq:weights} so it is in $\calW(G)$
and the  observation follows.

To see that $|\Omega_{S}| = |\Omega_\emptyset|$, we establish a bijection~$\tau$ 
between $\Omega_S$ and $\Omega_\emptyset$.
Let $S=\{u,v\}$ and 
let $P$ be the set of edges in any fixed path from~$u$ to~$v$ in~$G$ --- such a path exists since $G$ is connected.
The bijection is straightforward. Given any 
  $A\in \Omega_S$, let $\tau(A) = A \oplus P$ and note that $\tau(A)\in \Omega_\emptyset$.
  \end{proof}

Now 
consider the weighting $(\lambda^{[i]},w^{[i]})$.
If $i<t$ then, for every size-$2$ subset~$S$ of~$V$,
we define  $w^{[i+1]}_S$ 
by running the weighted worm process,
as follows.
Set  $\epsilon=1/8$ and set $\delta^* = \delta/(n^2 t)$.  
Now run the
weighted worm process with weighting $(\lambda^{[i]},w^{[i]})$ to obtain 
(by Lemma~\ref{lem:mcmc})
an estimate $\hatR_S^{[i]}$ which, with probability at least $1-\delta^*$, satisfies

$$e^{-\epsilon} \hatR_S^{[i]} \leq \frac{
\hatZp{\emptyset}{G;\lambda^{[i]},w^{[i]}}   
  }{\hatZp{S}{G;\lambda^{[i]},w^{[i]}}   } \leq e^{\epsilon} \hatR_S^{[i]}.$$

In the proof of Lemma~\ref{lem:learnweights}, we will use Lemma~\ref{lem:mcmc} to
account for how long this run of the weighted worm process takes.
To conclude with the definition of $w^{[i+1]}_S$, let  $w_S^{[i+1]} = w_S^{[i]} \hatR_S^{[i]}$. 

\begin{lemma}
Assuming that 
the algorithm from Lemma~\ref{lem:mcmc} does not
fail when it is called to learn $\hatR_S^{0},\ldots,\hatR_S^{[i]}$,
 The weighting
$(\lambda^{[i+1]},w^{[i+1]})$ is in $\calW(G)$. 
\end{lemma}
\begin{proof}
 Consider any $S\subseteq V$ with 
 $|S|=2$.
 Then
 \begin{eqnarray*}
 \frac{\hatZp{S}{G;\lambda^{[i+1]},w^{[i+1]}}}
 {\hatZp\emptyset{G;\lambda^{[i+1]},w^{[i+1]}}} &= 
  \frac {\sum_{A \in \Omega_S}  w^{[i+1]}_S \prod_{e\in A} \lambda^{[i+1]}(e)  \strut}
{\strut\sum_{A \in \Omega_\emptyset}  w^{[i+1]}_\emptyset \prod_{e\in A} \lambda^{[i+1]} (e) }\\
&= 
  \frac {w^{[i]}_S  \hatR_S^{[i]} \sum_{A \in \Omega_S}   \prod_{e\in A} \lambda^{[i+1]} (e) \strut}
{\strut   \sum_{A \in \Omega_\emptyset}   \prod_{e\in A} \lambda^{[i+1]} (e) }.
\end{eqnarray*}

Using the upper bound on  $\hatR_S^{[i]}$ and
$\frac{\lambda^{[i]}(e)}{(1+\bd)} \leq \lambda^{[i+1]}(e) \leq \lambda^{[i]}(e)$,
this
 quantity is at most 
\begin{align*}
   \frac {w^{[i]}_S 
   e^\acc \hatZp{\emptyset}{G;\lambda^{[i]},w^{[i]}}   
    \sum_{A \in \Omega_S}   \prod_{e\in A} \lambda^{[i]}(e) }
{  
\hatZp{S}{G;\lambda^{[i]},w^{[i]}}
\sum_{A \in \Omega_\emptyset}   \prod_{e\in A} \frac{\lambda^{[i]}(e)}{1+\bd} }
&\leq e^\acc {(1+\bd)}^m
  \frac {w^{[i]}_S 
    \hatZp\emptyset{G;\lambda^{[i]},w^i}   
    \sum_{A \in \Omega_S}   \prod_{e\in A} \lambda^{[i]}(e) }
{ 
\hatZp{S}{G;\lambda^{[i]},w^i}
\sum_{A \in \Omega_\emptyset}   \prod_{e\in A}  {\lambda^{[i]}(e)} }\\
&= e^\acc {(1+\bd)}^m \leq 2.\end{align*}
 Similarly, the quantity is at least~$1/2$.
\end{proof}

Collecting what we have done in this section, we get the following lemma.  
\begin{lemma}
\label{lem:learnweights} 
There is a randomised algorithm that takes as input
a connected graph
$G=(V,E)$   with $|V|=n$ and $|E|=m$ and an edge weighting
$\blambda:E \to \ourrange\cap (0,1)$.
The algorithm also takes a failure probability~$\delta$. With probability at least $1-\delta$, it computes
a subset weighting $w$ of $G$ 
so that $(\blambda,w)$ is in  $\calW(G)$.
Let $\lambdamin=\min_{e\in E}\lambda(e)$.
The running time of the algorithm is at most  a polynomial in $n$,   $1/\lambdamin$, and $\log(1/\delta)$.
\end{lemma}
\begin{proof}

Let $t = \max_{e\in E} \left\lceil \log(1/\lambda(e)) / \log(1+\bd)\right\rceil$.
Note that $t = O(m  \log (1/\lambdamin))$.
The algorithm constructs the sequence $(\lambda^{[0]},w^{[0]}),\ldots,(\lambda^{[t]},w^{[t]})$ of weightings
as described in this section where $\lambda^{[t]} = \lambda$.

We just have to collect the failure probabilities and running times.
As note earlier, for $i\in \{0,\ldots,t-1\}$, for each size-$2$ subset~$S$ of~$V$,
we estimate $\hatR_S^{[i]}$ using the weighted worm process (Lemma~\ref{lem:mcmc}) with $\epsilon=1/8$
and specified failure probability $\delta^* = \delta/(n^2 t)$.

The running time  from Lemma~\ref{lem:mcmc} is at most 
a polynomial in~$n$, $1/\lambdaimin$,  
and $\log(1/\delta^*)$.
Recall that $\lambda^{[i]}(e) \geq \lambda(e)$.
Thus, the overall running time is at most a polynomial in $n$,   $1/\lambdamin$, and $\log(1/\delta)$.
By a union bound, the overall failure probability is at most~$\delta$.
\end{proof}

\section{The Proof of Theorem~\ref{thm:main}}\label{sec:proof}

\begin{thmmain}
\statethmmain
\end{thmmain}

\begin{proof}
We start by reviewing what are the inputs and outputs of an FPRAS for $\FerroIsingCorr$.

The input consists of  
an input to $\FerroIsingCorr$,
 an accuracy parameter $\epsilon\in(0,1)$, and
 a failure probability $\delta\in (0,1)$.
An input to $\FerroIsingCorr$ consists of 
 a   graph $G=(V,E)$  with 
specified vertices $s$ and $t$ and an edge weighting $\bbeta:E\to \ferrorange$ of $G$.
Let $n=|V|$.
We need to be more specific about how the edge weighting~$\beta$ is represented.
Recall from the introduction that
each weight $\beta(e)$ satisfies $\beta(e)>1$ and  is represented in the input by two
positive integers $P(e)$ and $Q(e)$ (specified in unary in the input)
such that  
 and $\beta(e) = 1+P(e)/Q(e)$.

With probability at least $1-\delta$, the
output $\widehat{C}$ of the FPRAS should satisfy
\begin{equation}
\label{finalgoal} e^{-\epsilon} \widehat{C} \leq \Ex_{\piI}[\sigma(s)\sigma(t)] \leq e^{\epsilon} \widehat{C}.\end{equation}

Finally, in order to be an FPRAS, the running time should be at
most  a polynomial in
$n$,  $\sum_{e\in E} (P(e)+Q(e))$, $1/\epsilon$, and
$\log(1/\delta)$.

 If $s$ and $t$ are in different connected components of~$G$ then
$\Ex_{\piI}[\sigma(s)\sigma(t)]$=0, so  we can just output~$0$ in this case.
If $s$ and $t$ are in the same connected component, $G'$, of~$G$
then $\Ex_{\piI}[\sigma(s)\sigma(t)] = \Ex_{\piIprime}[\sigma(s)\sigma(t)] $.
So assume without loss of generality, for the rest of the proof, that $G$ is connected.

 Now  let $\blambda$ be the edge-weighting of~$G$ defined by
$\blambda(e)=(\bbeta(e)-1)/(\bbeta(e)+1)$.  
Let $\lambdamin=\min_{e\in E}\lambda(e)$.

The FPRAS should first
run the algorithm of Lemma~\ref{lem:learnweights} with input $G$, $\blambda$ and $\delta/2$.
With probability at least $1-\delta/2$, it 
computes
a subset weighting $w$ of $G$ 
so that $(\blambda,w)$ is in  $\calW(G)$.

Let $S=\{s,t\}$.
Suppose that the algorithm of  Lemma~\ref{lem:learnweights}
has succeeded. Recall from Lemma~\ref{lem:worm}
that $ \Ex_{\piI}[\sigma(s)\sigma(t)] =\frac{\ZE{S}}{\ZE\emptyset}$.
Also, by plugging in the definitions of $\hatZS{S}$ (from the beginning of Section~\ref{sec:weightedworm})
and $\ZE{S}$ (from the beginning of Section~\ref{sec:even}) we have
$$\Ex_{\piI}[\sigma(s)\sigma(t)] =\frac{\ZE{S}}{\ZE\emptyset} = \frac{\hatZS{ S} }{w_S\hatZS{\emptyset} }.$$
Since we already know $w_S$, the goal is to compute a quantity~$\hatQ$
such that
$$e^{-\epsilon} \hatQ \leq \frac{\hatZS{ S} }{\hatZS{\emptyset} } \leq e^{\epsilon} \hatQ.$$
Then we  satisfy~\eqref{finalgoal} by taking $\widehat{C} = \hatQ/w_S$.

The estimate $\hatQ$ can be obatined  by  running
the algorithm of Lemma~\ref{lem:mcmc}
with input $G$, weighting $(\blambda,w)$, set $S = \{s,t\}$,
accuracy parameter $\epsilon$,   desired failure probability $\delta^*=\delta/2$,
letting $\hatR$ be the output of the algorithm, and taking  $\hatQ = 1/\hatR$.

The running time of both algorithms is   at most a polynomial in
$n$, $\log(1/\delta)$, $1/\epsilon$ and  
$$1/\lambdamin = \max_{e\in E} \frac{\beta(e)+1}{\beta(e)-1} 
=\max_{e\in E} \left\{1 + \frac{2}{\beta(e)-1}\right\}
=\max_{e\in E} \left\{1 + \frac{2 Q(e)}{ P(e)}\right\}  .$$
\end{proof} 
 
\begin{remark}\label{rem:unarybinary}
It is possible to improve the run-time of the algorithm of Theorem~\ref{thm:main}
so that the dependence on 
$1/\lambdamin =  \max_{e\in E} \left\{1 + \frac{2}{\beta(e)-1}\right\}$ is logarithmic, rather than polynomial.
To do this, we pre-process the graph~$G$.
If an edge~$e$ has a weight $\beta(e)$ that is very close to~$1$
then it is replaced with a subgraph~$J$. The weights of the edges of $J$ are constants
bounded above~$1$, but the  overall effect of~$J$ is to simulate the weight $\beta(e)$ with exponential precision.
The technical details of the simulation are  very similar to what we   do in Lemma~\ref{lem:J}
of Section~\ref{sec:antiferro}.
We omit the details since   polynomial-time algorithms 
(as opposed to strongly polynomial-time algorithms) are sufficient for our purposes.
 \end{remark}

\section{The Antiferromagnetic Case}\label{sec:antiferro}

Consider the  antiferromagnetic Ising model  
on a graph $G=(V,E)$ 
as defined in Section~\ref{sec:intro}.
We will consider the situation where,
for some $\myb \in (0,1)$,
$G=(V,E)$ is a graph with an edge weighting~$\bbeta$
that assigns value $\bbeta(e) = \myb$ to every $e\in E$.
We will simplify the notation by  
defining 
$$\wtMonoI(\sigma)=  \myb^{ | \{ e=\{u,v\}\in E: \sigma(u)=\sigma(v) \}| } $$
with the corresponding partition function 
$\ZMonoI=\sum_{\sigma:V\to\{-1,+1\}}\wtMonoI(\sigma)$
and Gibbs distribution
$\piMonoI{G}(\cdot)$.
We will be interested in the following computational problem, with parameter $\myb\in(0,1)$.

\prob
{$\Correlation$.} 
{ A   graph $G$ with specified vertices $s$ and $t$} 
{ A correct statement of the form  ``$\Ex_{\piMonoI{G}}[\sigma(s)\sigma(t)]\geq 0$''
or ``$\Ex_{\piMonoI{G}}[\sigma(s)\sigma(t)]\leq 0$''.}

The purpose of this section is to prove Theorem~\ref{thm:anti}
which  states that, for any $\myb\in(0,1)$, $\Correlation$ is \#P-hard.
Note that the
\#P-hardness does not come from the difficulty of  
determining whether or not 
$\Ex_{\piMonoI{G}}[\sigma(s)\sigma(t)] $  is zero --- 
an algorithm for $\Correlation$ is allowed to give either answer in this case.
Theorem~\ref{thm:anti}
implies that it is \#P-hard to approximate 
 $\Ex_{\piMonoI{G}}[\sigma(s)\sigma(t)] $
 within any specified factor, since such an approximation would allow
 one to determine  either
 $\Ex_{\piMonoI{G}}[\sigma(s)\sigma(t)]\geq 0$
or $\Ex_{\piMonoI{G}}[\sigma(s)\sigma(t)]\leq 0$.

We start with some notation.
Given vertices~$s$ and~$t$ of~$G$, 
let $\Psi_{s+,t+}$ be the set of assignments 
$\sigma\colon V\to\{-1,+1\}$ that satisfy $\sigma(s)=+1$ and $\sigma(t)=+1$. Let
$$\ZIstGb{s+}{t+}{G}{\sbeta}  = \sum_{\sigma \in \Psi_{s+,t+} }
\wt(\sigma).$$
Define the sets of assignments
$\Psi_{s+,t-}$, $\Psi_{s-,t+}$, $\Psi_{s-,t-}$
and the partition functions
$\ZIstGb{s+}{t-}{G}{\sbeta} $,
$ \ZIstGb{s-}{t+}{G}{\sbeta} $, and
$\ZIstGb{s-}{t-}{G}{\sbeta}$  similarly.
We use the following notation of implementation.
\begin{definition}
A graph $J$ is said to \emph{$\myb$-implement}
a rational number $\myb'$  
if there are vertices $s$ and $t$ of $J$ such that
 $
 {\ZIstGb{s+}{t+}{J}{\sbeta}}/{\ZIstGb{s+}{t-}{G}{\sbeta}}  = \myb'   $.
We call $s$ and $t$ the terminals of $J$.
\end{definition}
We will use the following lemma for implementations.

 \begin{lemma}\label{lem:J}
Fix $\myb\in (0,1)$. 
There is a polynomial-time algorithm that takes as input 
\begin{itemize}
\item A  positive integer~$n$, in unary,
\item  a target-edge weight  $\myb'\in [\myb^n,\myb^{-n}] $, and
\item    a rational accuracy parameter~$\acc\in (0,1)$, in binary.
\end{itemize}
The algorithm produces a graph $J$ with 
terminals $s$ and $t$ that $b$-implements 
a value $\goal$ satisfying
$|\goal - \myb'| \leq \acc$. 
The  size of~$J$
  is at most a polynomial in $n$ and $\log(1/\acc)$, independently of~$\myb'$. 
    \end{lemma}

 \begin{proof} 
If $\myb'=1$ then $J$ is the graph with vertices $s$ and $t$ and no edges.
So suppose $\myb'\neq 1$.

Let $P_\ell$ be an $\ell$-edge path with endpoints $s$ and $t$.
Let $f_\ell = 
\ZIstGb{s+}{t+}{P_\ell}{\sbeta}$ and 
$a_\ell = 
\ZIstGb{s+}{t-}{P_\ell}{\sbeta}$. Then $f_1=\sbeta$, $a_1=1$
and we have the  system
$f_\ell = \sbeta f_{\ell-1} + a_{\ell-1}$
and $a_\ell = f_{\ell-1} + \sbeta a_{\ell-1}$.
Thus, $f_\ell = (1/2) ((\sbeta+1)^\ell + (\sbeta-1)^\ell)$ and
$a_\ell = (1/2) ((\sbeta+1)^\ell - (\sbeta-1)^\ell)$.
 Let $\zeta_\ell = f_\ell/a_\ell$. Then
$$  \zeta_\ell = 
\frac{f_\ell}{a_\ell} 
= 
{ \frac{(\sbeta+1)^\ell + (\sbeta-1)^\ell}
{(\sbeta+1)^\ell - (\sbeta-1)^\ell} } = 
  { 1 + \frac{2  }{    c^\ell -  1} },$$  
where $c=( \sbeta+1)/(\sbeta-1) <-1 $.  
Note that for odd~$\ell$ the values of $\zeta_\ell$ are in $(0,1)$ and they increase.
Also, for even~$\ell$ the values of $\zeta_\ell$ are greater than~$1$ and they decrease.
  
The graph $J$ is constructed by combining  copies of $P_\ell$ (for different values of~$\ell$), identifying
the vertex $s$ in all copies, and identifying the vertex $t$ in all copies.  
Let  
$$L = \left\lceil \frac{\log\left(\frac{2 }{b^n\acc} + 1\right)}{\log(c^2)} \right\rceil .$$
We will use paths $P_\ell$ with $\ell \leq 2 L + 1$.
The value of $L$ is defined so that $L$ is 
at most a polynomial in 
$n$ and
$\log(1/\acc)$, as required in the statement of the lemma, and
also  
\begin{equation}\label{eq:generic}
c^{2L} \geq {2 b^{-n}}/{\acc} + 1.
\end{equation}
  
The graph $J$ is constructed as follows.
\begin{itemize}
\item If $\myb'>1$: 
Set $B_0=\myb'>1$.
For odd $j\in \{2,\ldots,2L+1\}$, let   $d_j=0$.
For even  $j\in \{2,\ldots,2L+1\}$, 
let $d_j$ be the largest non-negative integer 
such that $\zeta_j^{d_j} \leq  B_{j-2}$
and let $B_j = B_{j-2}/\zeta_j^{d_j}\geq 1$.

\item If $\myb'<1$:
Set $B_1=\myb'<1$.
For even $j\in \{2,\ldots,2L+1\}$, let   $d_j=0$.
For odd  $j\in \{2,\ldots,2L+1\}$, 
let $d_j$ be the largest non-negative integer  
such that $\zeta_j^{d_j} \geq B_{j-2}$
and let $B_j = B_{j-2}/\zeta_j^{d_j}\leq 1$.
\end{itemize}  
The graph $J$ is constructed by taking $d_j$ copies of $P_j$ for $j\in \{2,\ldots,2L+1\}$,
identifying the vertex~$s$ 
in all copies and identifying the vertex~$t$ in all copies.
The value that $J$ $b$-implements is
$$
\goal = \frac{\ZIstGb{s+}{t+}{J}{\sbeta}}{
\ZIstGb{s+}{t-}{J}{\sbeta}} = \prod_{j=2}^{2L+1} \zeta_j^{d_j}$$
  
We next show that $|\goal - \myb'| \leq \acc$.
\begin{itemize}
\item If $\myb'>1$: The construction guarantees
$ \frac{\myb'}{\zeta_{2L}}\leq  \goal \leq \myb' $.
But~\eqref{eq:generic} implies
 $\zeta_{2L} \leq 1 + \frac{\acc}{b^{-n}}\leq 1 + \frac{\acc}{\myb'} \leq \frac{1}{1-\frac{\acc}{\myb'}}$,
  so   $\myb'-\acc \leq \myb'/\zeta_{2L}$.

\item If $\myb'<1$: The construction guarantees
$\myb' \leq  \goal \leq \frac{\myb'}{\zeta_{2L+1}} $.
But~\eqref{eq:generic} implies
${|c|}^{2L+1} \geq c^{2L} \geq {2 b^{-n}}/{\acc} + 1$
so 
$\zeta_{2L+1} \geq 1 - \frac{2}{{|c|}^{2L+1}+1}\geq 1- \frac{\epsilon}{b^{-n}+\epsilon} \geq 1-\frac{\epsilon}{b'+\epsilon}=\frac{b'}{b'+\epsilon}$,
so
$ {\myb'}/{\zeta_{2L+1}} \leq \myb'+\epsilon$.
 
\end{itemize}

To finish the bound on the size of $J$, we will show that $d_2$ and $d_3$ are $O(n)$
and that, for every $j\in \{4,\ldots, 2L+1\}$, $d_j = O(1)$.
First, $d_2 \leq \log_{\zeta_2}(\myb')$
 where $\zeta_2 = (b^2+1)/(2 b)$.
Also, $d_3 \leq  \log_{1/\zeta_3} (1/\myb')$ where
$1/\zeta_3 =  (1+3b^2)/(b(3+b^2))$.
   
Finally, let
$d = \lceil {c}^4/({c}^2-1)\rceil $.  
Note that we could replace ``$c$'' with ``$|c|$'' in the definition of~$d$ without changing the definition,
so, plugging the definition in, we find, for $j\geq 4$, that
${|c|}^{j}-1 \leq d ({|c|}^{j-2}-1)$.
This implies
\begin{equation}
\label{eq:last}
{\left(1+\frac{2}{|c|^j-1}\right)}^d \geq 1 + \frac{2d}{|c|^j-1} \geq 1+ \frac{2}{|c|^{j-2}-1}.
\end{equation}
If $j\geq 4$ is even then~\eqref{eq:last}
implies
$\zeta_j^d \geq \zeta_{j-2}$, so $d_j \leq d$.
If $j\geq 4$ is odd then~\eqref{eq:last}
gives 
$${\left(\frac{|c|^j-1}{|c|^j+1}
\right)}^d \leq \frac{|c|^{j-2}-1}{|c|^{j-2}+1},$$
and all numerators and denominators are negative, so multiplying them by $-1$
we get
$${\left(\frac{c^j+1}{c^j-1}
\right)}^d \leq \frac{c^{j-2}+1}{c^{j-2}-1},$$
so $\zeta_j^d \leq \zeta_{j-2}$ and $d_j \leq d$.
\end{proof}

 \begin{theorem}\label{thm:anti}
 Let $\myb \in (0,1)$ 
 be a rational number. 
 Then $\Correlation$ is  $\mathrm{\#P}$-hard. \end{theorem}
 
\begin{proof}

Fix $\myb\in (0,1)$.
We will show
how to use an oracle for $\Correlation$
to give a polynomial time algorithm for exactly 
computing $\ZMonoI$, a problem that is known to be \#P-hard
(see \cite[Theorem 14]{JS93} for \#P-hardness of a multi-variate version and \cite[Corollary 2]{VertiganWelsh} 
for a result that implies \#P-hardness of the version considered here).

Let $G=(V,E)$ be a graph with $n$ vertices and 
$m>0$ edges.
We will show how to compute $\ZMonoI$
using the given oracle for $\Correlation$.

As we will see, the information provided by the oracle for $\Correlation$ can naturally be used to provide
a multiplicative approximation to $\ZMonoI$.
Since we need additive approximations in order to compute $\ZMonoI$ precisely,
we have to switch back and forth between additive and multiplicative
approximations. 
To this end, let $b=p/q$ for integers~$p$ and~$q$.
 Choose $m'= O(m)$ such that $b^{m'} \leq (1/q)^m$.

 Note that 
 $\ZMonoI = \sum_{j=0}^m b^j c_j$,
where $c_j$ is the number of configurations $\sigma\colon V\to\{-1,+1\}$
which induce $j$~edges with like spins in~$G$.
This implies that $\myb^m 2^n \leq \ZMonoI \leq 2^n$.
Now let $\delta = b^{m'} 2^{-(n+3)}$.
Suppose that $\widehat{Z}$ satisfies
\begin{equation}\label{eq:getmult}
e^{-\delta} \ZMonoI \leq \widehat{Z} \leq e^{\delta} \ZMonoI
\end{equation}
so that
$(1-\delta) \ZMonoI \leq \widehat{Z} \leq (1+2\delta) \ZMonoI$.
We conclude that
$$
|\widehat{Z} - \ZMonoI | \leq 2 \delta \ZMonoI \leq 2 \delta 2^n \leq \myb^{m'}/4
$$
so from $\widehat{Z}$ we learn $\ZMonoI$ precisely. To see this, note 
that any  interval of  length $b^{m'}/2$ contains the value of
 at most one   polynomial of the 
form $\sum_{j=0}^m b^j c_j$ with integer coefficients. 
Consider two such polynomials $Z_1 = \sum_{j=0}^m b^j c_j$ and $Z_2 =\sum_{j=0}^m b^j c'_j$,
both with integer coefficients. 
Set $a_j = c_j - c'_j$. Then 
$$Z_1 - Z_2 = \sum_{j=0}^m  \frac{a_j p^j }{ q^j} \leq \frac{\sum_{j=0}^m a_j p^j q^{m-j}}{q^m},$$ but
the numerator is an integer, so if $Z_1\neq Z_2$ then $|Z_1-Z_2| \geq 1/q^m \geq b^{m'}$,
so $Z_1$ and $Z_2$ cannot both be in an interval of length $b^{m'}/2$.

Thus, from now on, our goal will be to show how to
use the given oracle for $\Correlation$ to
obtain $\widehat{Z}$
satisfying~\eqref{eq:getmult}.  This will complete the proof of Theorem~\ref{thm:anti} and, a fortiori, Theorem~\ref{thm:negthm}.

Let the edges of $G$ be $e_1,\ldots,e_m$
and, for $j\in [m]$, let $G_j = (V,\{e_1,\ldots, e_j\})$. 
Denote the endpoints of $e_j$ by $s_j$ and $t_j$.
Using the notation from the beginning of the section, let
$\nu_j = \ZIstGb{s_j+}{t_j-}{G_{j-1}}{\sbeta}/\ZIstGb{s_j+}{t_j+}{G_{j-1}}{\sbeta}$
and let
$\alpha_j= (b+ \nu_j )/(1+ \nu_j)$.
Observe that $\ZMonoIG{G_j}=2(\myb \ZIstGb{s_j+}{t_j+}{G_{j-1}}{\sbeta}+\ZIstGb{s_j+}{t_j-}{G_{j-1}}{\sbeta})$  
and $\ZMonoIG{G_{j-1}}=2(  \ZIstGb{s_j+}{t_j+}{G_{j-1}}{\sbeta} +\ZIstGb{s_j+}{t_j-}{G_{j-1}}{\sbeta})$, and hence
$\ZMonoIG{G_j}=\alpha_j \ZMonoIG{G_{j-1}}$.
Therefore,
$$\ZMonoI = \ZMonoIG{G_m} = \left(\prod_{j=1}^m \alpha_j\right) \ZMonoIG{G_0} = 2^n \prod_{j=1}^m \alpha_j,$$
so to finish it suffices to show, for $j\in [m]$, that we can
use an oracle for $\Correlation$ to approximate $\alpha_j$ with  multiplicative error $\exp(\pm\delta/m)$.  

Suppose that we could produce $\hat{\nu}_j$
satisfying $|\hat{\nu}_j - \nu_j| \leq \myb \delta/(5 m)$. Then, setting
$\hat{\alpha}_j= (b+ \hat{\nu}_j )/(1+ \hat{\nu}_j)$, we have
 $\alpha_j \exp(-\delta/m)\leq  \hat{\alpha}_j \leq\alpha_j\exp(\delta/m)$, as required.

So, to finish it suffices to show, for $j\in [m]$, that we can
use an oracle for $\Correlation$ to approximate $\nu_j$ with additive error at most $\delta' = \myb \delta/(5 m)$. 
Our basic approach is the binary-search method that
the authors used in \cite{sign} to show that it is \#P-hard to compute the sign of the Tutte polynomial.

The invariant that we will maintain is that $\nu_j$ lies in an interval $[\nmin,\nmax]$.
We will repeatedly use the oracle to reduce the length of the interval by a constant factor, until
the length is at most $\delta'$ (in which case we can take $\hat{\nu_j}$ to be any point in the interval).
To initialise the search interval, 
we take $[\nmin,\nmax] = [b^n,b^{-n}]$.
It is clear that $\nu_j$ lies in this interval,    since flipping the spin at $t_j$ affects at most $n-1$ incident edges, and 
therefore changes the weight of  a configuration by a factor that is at least $b^n$ and at most $b^{-n}$.
 
Our basic approach is as follows.
Suppose that we can construct a graph $J$ with terminals $s$ and $t$ to $b$-implement
a point $\goal$ in the middle third of the interval.
 Let
$\bbeta$ be the edge-labelling of $G_j$ that assigns value $\goal$ to edge~$e_j$
and $\myb$ to all other edges. 
Let $G_j(J)$ be the graph formed from $G_j$ by replacing the edge~$e_j$ with the graph $J$
(identifying the terminal $s$ of $J$ with the vertex $s_j$ of $G_j$ and identifying the terminal $t$ of~$J$
with the vertex $t_j$ of $G_j$).
Then
$$ 
\Ex_{\piMonoIbeta{G_j(J)}{\myb}}[\sigma(s_j)\sigma(t_j)] = 
\Ex_{\piMonoIbeta{G_j}{\bbeta}}[\sigma(s_j)\sigma(t_j)] = 
\frac{
\goal \ZIstGb{s_j+}{t_j+}{G_{j-1}}{\sbeta}   -
\ZIstGb{s_j+}{t_j-}{G_{j-1}}{\sbeta} }
{\goal \ZIstGb{s_j+}{t_j+}{G_{j-1}}{\sbeta}   + 
\ZIstGb{s_j+}{t_j-}{G_{j-1}}{\sbeta}} .
$$
Using an oracle for $\Correlation$ we can 
determine either 
that this quantity is at least~$0$
(in which case $\goal \geq \nu_j$ and we can recurse on $[\nmin,\goal]$)
or that it is at most~$0$ (in which $\goal \leq \nu_j$
and we can recurse on $[\goal,\nmax]$).
Either way, the length of the new interval 
is at most~$2/3$ of the length that it was, so 
after $\text{O}(\log(1/b)(m'+n))$ iterations, the length of the interval will have shrunk to length at most $\delta'$, as required.

Finding the required~$J$ 
is straightforward --- this can be done by 
taking $\myb' = (\nmin+\nmax)/2$ to be the centre of the interval and
using Lemma~\ref{lem:J} with
inputs $n$, $\myb'$ and $\acc = \delta'/6$.
The size of $J$ is at most a polynomial in $n$ and $\log(1/\acc)$, which is polynomial in~$n$.
 \end{proof}

\bibliographystyle{plain}
\bibliography{\jobname}

\end{document}